\documentclass[letterpaper, 10 pt, conference]{ieeeconf}  % Comment this line out if you need a4paper

\IEEEoverridecommandlockouts                              % This command is only needed if 
\usepackage[utf8]{inputenc} % allow utf-8 input
\usepackage[T1]{fontenc}    % use 8-bit T1 fonts
\usepackage{hyperref}       % hyperlinks
\usepackage{url}            % simple URL typesetting
\usepackage{booktabs}       % professional-quality tables
\usepackage{amsfonts}       % blackboard math symbols
\usepackage{nicefrac}       % compact symbols for 1/2, etc.
\usepackage{microtype}      % microtypography
\usepackage{xcolor}         % colors

\usepackage{amsmath,bm}
\usepackage{mathtools}
\usepackage{centernot}
\usepackage{algorithm}
\usepackage{algorithmic}
\usepackage{comment}

\DeclareMathOperator*{\argmax}{arg\,max}

\newcommand{\E}{\operatorname{\mathbb E}}
\newcommand{\innermid}{\;\middle\lvert\;}

\newtheorem{lemma}{Lemma}

\newtheorem{corollary}{Corollary}
\newtheorem{theorem}{Theorem}

\newtheorem{definition}{Definition}
\newtheorem{assumption}{Assumption}

\newcommand\bc[1]{\left({#1}\right)}
\newcommand\cbc[1]{\left\{{#1}\right\}}
\newcommand\abs[1]{\left|{#1}\right|}
\newcommand\brk[1]{\left\lbrack{#1}\right\rbrack}
\newcommand\norm[1]{\left\|{#1}\right\|}

\newcommand\Erw{\mathbb{E}}

\allowdisplaybreaks

\title{\LARGE \bf
Mean Field Games on Weighted and Directed Graphs \\via Colored Digraphons
}

\author{Christian Fabian, Kai Cui and Heinz Koeppl% <-this % stops a space
\thanks{This work has been co-funded by the Hessian Ministry of Science and the Arts (HMWK) within the projects "The Third Wave of Artificial Intelligence - 3AI" and hessian.AI, and the LOEWE initiative (Hesse, Germany) within the emergenCITY center.}
\thanks{All authors are with the Department of Electrical Engineering and Information Technology
        Technische Universität Darmstadt,
        64287 Darmstadt, Germany
        {\tt\small \{christian.fabian, kai.cui,  heinz.koeppl\}@tu-darmstadt.de}}%
}

\begin{document}

\maketitle
\thispagestyle{empty}
\pagestyle{empty}

%%%%%%%%%%%%%%%%%%%%%%%%%%%%%%%%%%%%%%%%%%%%%%%%%%%%%%%%%%%%%%%%%%%%%%%%%%%%%%%%
\begin{abstract}
The field of multi-agent reinforcement learning (MARL) has made considerable progress towards controlling challenging multi-agent systems by employing various learning methods. Numerous of these approaches focus on empirical and algorithmic aspects of the MARL problems and lack a rigorous theoretical foundation. Graphon mean field games (GMFGs) on the other hand provide a scalable and mathematically well-founded approach to learning problems that involve a large number of connected agents. In standard GMFGs, the connections between agents are undirected, unweighted and invariant over time. Our paper introduces colored digraphon mean field games (CDMFGs) which allow for weighted and directed links between agents that are also adaptive over time. Thus, CDMFGs are able to model more complex connections than standard GMFGs. Besides a rigorous theoretical analysis including both existence and convergence guarantees, we provide a learning scheme and illustrate our findings with an epidemics model and a model of the systemic risk in financial markets.
\end{abstract}

%%%%%%%%%%%%%%%%%%%%%%%%%%%%%%%%%%%%%%%%%%%%%%%%%%%%%%%%%%%%%%%%%%%%%%%%%%%%%%%%

%%%%%%%%%%%%%%%%%%%%%%%%%%%%%%%%%%%%%%%%%%%%%%%%%%%%%%%%%%%%%%%%%%%%%%%%%%%%%%%%

\section{Introduction}

Over the last decades, Reinforcement Learning (RL) \cite{sutton2018reinforcement} has become a key method for solving problems across various research areas, including  fields such as robotics \cite{kormushev2013reinforcement, polydoros2017survey}, communication networks \cite{qian2019survey}, and biology \cite{mahmud2018applications}. Despite these advances, many open problems remain, especially in the domain of multi-agent reinforcement learning (MARL).  
Broadly speaking, MARL \cite{busoniu2008comprehensive} is concerned with problems that involve multiple agents interacting with each other. As the number of agents grows, many of the existing solution schemes for such models become intractable. Beside these practical algorithmic issues, numerous MARL approaches focus on empirical aspects but do not provide a well-founded, mathematical analysis of the underlying methods.

A concept that aims to resolve both of the just mentioned challenges is the one of mean field games (MFG) which was introduced independently by \cite{huang2006large} and \cite{lasry2007mean}. MFGs, in general, model games where a large number of indistinguishable agents interact with each other such that each single agent has a negligible influence on the other agents. The crucial idea of MFGs is that each individual only interacts with the other agents through the mean field of agents, i.e. a probability distribution over some state space.

Although MFGs facilitate the mathematical and algorithmic solution of many otherwise intractable models (see \cite{caines2021mean} for an overview), their initial assumptions appear to be too rigid for various applications. This especially includes the fact that in many real world scenarios, individuals do not interact with the entity of all other agents but with a rather small subgroup. To mathematically model such more complicated setups, one can utilize graphon mean field games (GMFGs) \cite{caines2018graphon,caines2019graphon}.

Intuitively, GMFGs are similar to MFGs but allow agents to be connected by a graph such that only neighbors influence each other. This enables a more realistic analysis of crucial problems such as modelling the dynamics of an epidemic \cite{vizuete2020graphon}. A shortcoming of this technique is, that the underlying standard graphons \cite{lovasz2012large} are only suitable limiting objects for undirected and unweighted graphs. Therefore, scenarios that contain directed interactions, e.g. sending and receiving messages in a communication network \cite{zhang2011lyapunov}, or weighted interactions such as monetary transactions in a financial market \cite{kenett2012dependency}, cannot be depicted properly by standard GMFGs.

To address these challenges, our paper proposes the concept of colored digraphon mean field games (CDMFGs). We build our modelling approach on the mathematical concept of colored digraphons \cite{lovasz2013non} which are a suitable limiting object for sequences of weighted and directed graphs. Furthermore, our examples also outline how CDMFGs are able to depict adaptive connections between agents, i.e. when neighbors of an agent vary over time.
Our contributions can be summarized as follows: i) We introduce the novel concept of CDMFGs which provides a modelling framework for MFGs on weighted, directed, and adaptive graphs; ii) we conduct a rigorous theoretical analysis of CDMFGs and give both existence and convergence guarantees; iii) we provide a MARL learning scheme to determine equilibria in empirical applications; and iv) we demonstrate the empirical capabilities of our approach by providing examples on the systemic risk in financial markets and the spread of contagious diseases.

\section{Colored Digraphons}

This section provides a short introduction to $k$-colored digraphons. For more details, see e.g. \cite{lovasz2013non}.

The general idea of graphons \cite{lovasz2012large} is to provide a limiting object for growing sequences of graphs. Then, one can avoid large adjacency matrices and instead work with the more amenable graphons which contain the asymptotic structure of the graph sequence. A standard graphon $W$ is formally defined as a symmetric, measurable function $W \colon [0,1]^2 \to [0,1]$ which yields the limit for sequences of undirected and unweighted graphs. In order to depict directed graphs, often referred to as digraphs, we drop the symmetry condition which yields measurable functions $W \colon [0,1]^2 \to [0,1]$ called digraphons. To describe $k$-colored digraphs, i.e. digraphs where each edge is assigned one of $k$ different colors, we extend the previous definition to a tuple of $k$ digraphons $\mathbf{W} = (W^1, \ldots, W^k)$ such that $\sum_{i = 1}^k W^i = 1$ and where $\mathbf{W}$ is a called a $k$-digraphon. Note that as usual, the $k$ edge colors can be equivalently thought of as edge weights.

Based on these definitions, we can identify any $k$-colored digraph $G$ with $N$ nodes with a $k$-digraphon as follows. Let $I_1, \ldots, I_N$ be intervals of equal length that form a partition of $[0,1]$. Then, the graph $G$ can be associated with a $k$-digraphon $\mathbf{W}_G = (W_G^1, \ldots, W_G^k)$ defined by $W^h_G (x,y)= 1$ if $(x,y) \in I_i \times I_j$, $h \in \cbc{1, \ldots, k}$ and $h$ is the color of edge $i j$, and $W^h_G (x,y)= 0$ otherwise.

On the other hand, we can sample finite, random $k$-colored digraphs from a given $k$-digraphon $\mathbf{W}$. Assuming that our sampled graph is supposed to have $N$ nodes, we draw points $x_1, \ldots, x_N$ uniformly and independently from $[0,1]$ at random and interpret each value as a node. Subsequently, we color the edge between vertices $1 \leq i, j \leq N$ with color $h$ with probability $W^h (x_i, x_j)$, independently of all other edges. 
Formally, one can show that for every convergent sequence of $k$-colored digraphs there exists a limiting $k$-digraphon, see \cite[Proposition 2.1]{lovasz2013non}.

For comparing $k$-digraphons, we require a suitable norm. We start by recalling the cut norm for standard graphons $W\colon [0,1]^2 \to [0,1]$, i.e. 
\begin{align}
\norm{W}_\square \coloneqq \sup_{S,T \subseteq [0,1]} \abs{\int_{S \times T} W(x,y) \, \mathrm dx \, \mathrm dy},
\end{align}
where $S$ and $T$ range over the measurable subsets of $[0,1]$. 
The cut norm can be extended to  $k$-digraphons, i.e. for a $k$-digraphon $\mathbf{W} = (W^1, \ldots, W^k)$ we define
\begin{align}
    \norm{\mathbf{W}}_\square \coloneqq \sum_{h=1}^k \norm{W^h}_\square.
\end{align}
For the proofs in this paper, a different formulation of the above mentioned convergence concepts is useful. For a standard graphon $W\colon [0,1]^2 \to [0,1]$ and a converging sequence of standard graphons $(W_N)_N$ with $\norm{W - W_N}_\square \to 0$ we recall that this is equivalent to convergence in the operator norm $\norm{\cdot}_{L_\infty \to L_1}$, defined by
\begin{align}
    &\norm{W - W_N}_{L_\infty \to L_1} \\
    &\quad = \sup_{\norm{f}_\infty \leq 1} \int_\mathcal{I} \abs{\int_\mathcal{I} (W(x, y) - W_N(x, y)) f(y) \mathrm d y} \mathrm d x \to 0 \nonumber,
\end{align}
see \cite[Lemma 8.11]{lovasz2012large} for details. 

\begin{assumption} \label{ass:W}
The sequence of k-digraphons $(\mathbf{W}_{G_N})_{N \in \mathbb N}$ generated by the graph sequence $(G_N)_{N \in \mathbb N}$ converges in cut norm $\left\Vert \cdot \right\Vert_{\square}$  as $N \to \infty$ to some k-digraphon $\mathbf{W}$, i.e.
\begin{align} \label{eq:Wconv}
    \sum_{h=1}^k \left\Vert  W_{G_N}^h - W^h \right\Vert_{\square} \to 0
\end{align}
which implies 
$\sum_{h=1}^k \left\Vert W_{G_N}^h - W^h \right\Vert_{L_\infty \to L_1} \to 0$.
\end{assumption}

\section{Finite Agent Model}

In the following, if $Y$ is some finite set, $\mathcal{P} (Y)$ denotes the set of all probability measures on $Y$ and $\mathcal{B} (Y)$ the set of all bounded measures. 
The finite agent model (or $N$ agent model) consists of $N$ agents who can choose actions from some finite action set $\mathcal{U}$ while they are in one of the states defined by the finite state space $\mathcal{X}$. Furthermore, the game is setup on some $k$-colored digraph $G_N = (V_N, E_N)$ with $N$ nodes where the vertices represent agents while the (weighted and directed) edges indicate connections between pairs of agents. The game starts at time $t=0$ and ends at the terminal time $T \in \mathbb{N}$ where agents are allowed to act at each of the time points $t \in \mathcal{T} \coloneqq \{0, \ldots, T-1\}$. Formally, agent $i$ can choose to implement a policy $\pi^i \in \Pi \coloneqq \mathcal{P} ({\mathcal{U}})^{\mathcal{X} \times \mathcal{T}}$ where $\pi_t^i$ denotes the policy of agent $i$ at time $t$.

In contrast to models on simple graphs, a finite agent model on a $k$-colored digraph includes $2 k$ unnormalized neighborhood state distributions for each agent $i$, i.e. 
\begin{align}
	\mathbb G^i_{t, h, \textrm{out}} &\coloneqq \frac 1 {N}  \sum_{j \in V_N} \boldsymbol{1}_{\cbc{i j \textrm{ has color } h}} \delta_{X^j_t} \\
	\mathbb G^i_{t, h, \textrm{in}} &\coloneqq \frac 1 {N} \sum_{j \in V_N} \boldsymbol{1}_{\cbc{j i \textrm{ has color } h}} \delta_{X^j_t}
\end{align}
where $h$ can be any of the $k$ colors of the graph and $\delta$ denotes the Dirac measure. For notational convenience, we additionally define $\mathbb{G}_t^i \coloneqq (\mathbb G^i_{t, 1, \textrm{out}}, \ldots, \mathbb G^i_{t, k, \textrm{out}}, \mathbb G^i_{t, 1, \textrm{in}}, \ldots, \mathbb G^i_{t, k, \textrm{in}})$ such that $\mathbb{G}_t^i \in \mathcal{B} (\mathcal{X})^{2k}$.
Agents follow policies that only depend on their current state, i.e.  local Markovian feedback policies. This is taken into account by defining the model dynamics
\begin{align}
U^i_t \sim \pi^i_t ( \cdot \mid X^i_t) \quad \textrm{and}\quad X^i_{t+1} \sim P ( \cdot \mid X_t^i, U^i_t, \mathbb{G}^i_t)
\end{align}
for all $t \leq T, i \in V_N$ and $X^i_0 \sim \mu_0$ where $\mu_0$ is the initial state distribution and $P: \mathcal{X} \times \mathcal{U} \times \mathcal{B} (\mathcal{X})^{2k} \to \mathcal{P} (\mathcal{X})$ is some transition kernel. Finally, each agent's rewards are given by the function $r: \mathcal{X} \times \mathcal{U} \times \mathcal{B} (\mathcal{X})^{2k} \to \mathbb{R}$. Agents aim to competitively maximize their total expected rewards given by $J^i_N (\pi^1, \ldots, \pi^N) \coloneqq \E \left[ \sum_{t \in \mathcal{T}} r(X^i_t, U^i_t, \mathbb{G}^i_t)\right]$.

To analyze the above model, we recall the $(\epsilon, p)$-Markov-Nash equilibrium ($(\epsilon, p)$-MNE) concept, see e.g. \cite{elie2020convergence, carmona2004nash}. 
\begin{definition}
For $\epsilon, p > 0$, a $(\epsilon, p)$-MNE is a tuple $\boldsymbol{\pi} = (\pi^1, \ldots, \pi^N) \in \Pi^N$ such that $\forall i \in V'_N \subseteq V_N$
\begin{align*}
    \sup_{\pi \in \Pi} J_i^N \bc{\pi^1, \ldots,  \pi^{i -1}, \pi,  \pi^{i+1}, \ldots,  \pi^N} - J_i^N (\boldsymbol{\pi}) \leq \epsilon
\end{align*}
for some subset $V'_N$ of nodes with $\vert V'_N \vert \geq \lfloor (1-p) N \rfloor$.
\end{definition}

The advantage of the $(\epsilon, p)$-MNE is that it allows relatively small groups of agents to be out-of-equilibrium. Even in the asymptotic case, the network can contain local structures that deviate from the colored digraphon. However, taking $p \to 0$, brings us arbitrarily close to an equilibrium for all agents.

\section{Colored Digraphon Mean Field Game}

In CDMFGs, an infinite number of agents, i.e. the mean field, try to competitively maximize their expected rewards. In contrast to standard GMFGs, CDMFGs allow for directed and weighted connections between agents. This, in turn, also enables the modelling of adaptive network structures where connections change over time, which we highlight in the experiments section. CDMFGs are closely related to the finite agent models defined above and provide an increasingly accurate approximation for the $N$ agent case as $N$ increases, for details see the next sections.

The mean field of agents is represented by the interval $\mathcal{I} \coloneqq [0,1]$ where each $\alpha \in \mathcal{I}$ represents one agent. Then, we denote the space of measurable mean field ensembles by $\boldsymbol{\mathcal{M}} \coloneqq \mathcal{P}(\mathcal{X})^{\mathcal{I} \times \mathcal{T}}$ such that for any $(t, x, \boldsymbol{\mu}) \in \mathcal{T} \times \mathcal{X} \times \boldsymbol{\mathcal{M}}$ the map $\alpha \mapsto \mu_t^\alpha (x)$ is measurable. Similarly, $\boldsymbol{\Pi} \subseteq \Pi^{\mathcal{I}}$ is the space of measurable policy ensembles where measurable refers to $\alpha \mapsto \pi_t^\alpha (u \mid x)$ being measurable for all $(t, x, u, \boldsymbol{\pi}) \in \mathcal{T} \times \mathcal{X} \times \mathcal{U} \times \boldsymbol{\Pi}$.

\begin{assumption}\label{ass:pi_lipschitz}
$\boldsymbol{\pi} \in \boldsymbol{\Pi}$ is Lipschitz continuous up to a finite number of discontinuities.
\end{assumption}

Similar to the $N$ agent case, a CDMFG on a $k$-colored digraphon has $2 k$ unnormalized neighborhood state distributions for each agent $\alpha$ which are given by
\begin{align}
	\mathbb G^{\alpha, \boldsymbol{\mu}, \boldsymbol{W}}_{t, h, \textrm{out}} &\coloneqq \int_{\mathcal I} W^h(\alpha, \beta) \mu^\beta_t \, \mathrm d\beta\\
	\mathbb G^{\alpha, \boldsymbol{\mu}, \boldsymbol{W}}_{t, h, \textrm{in}} &\coloneqq  \int_{\mathcal I} W^h(\beta, \alpha) \mu^\beta_t \, \mathrm d\beta
\end{align}
where $h$ can be any of the $k$ colors.   As before, we let $\mathbb{G}_t^{\alpha, \boldsymbol{\mu}, \boldsymbol{W}}$ denote the vector of the $2 k$ neighborhood state distributions. We usually drop the dependence on $\boldsymbol{\mu}$ and $\boldsymbol{W}$ if it is clear from the context. Thus, the model dynamics are defined by
$U^\alpha_t \sim \pi^\alpha_t ( \cdot \vert X^\alpha_t)$ and $X^\alpha_{t+1} \sim P ( \cdot \vert X_t^\alpha, U^\alpha_t, \mathbb{G}^\alpha_t)$
with  $X_0^\alpha \sim \mu_0$. Agents try to competitively maximize their rewards $J_\alpha^{\boldsymbol{\mu}} (\pi^\alpha) \coloneqq \E \left[ \sum_{t \in \mathcal{T}} r(X^\alpha_t, U^\alpha_t, \mathbb{G}^\alpha_t)\right]$.

To state an equilibrium concept for CDMFGs, we require two functions $\Psi \colon \boldsymbol{\Pi} \to \boldsymbol{\mathcal{M}}$ and $\Phi \colon \boldsymbol{\mathcal{M}} \to 2^{\boldsymbol{\Pi}}$. $\Phi$ maps $\boldsymbol{\mu} \in \boldsymbol{\mathcal{M}}$ to the set of policies with $\pi^\alpha = \argmax_{\pi \in \Pi} J^{\boldsymbol \mu}_{\alpha}\bc{\pi^\alpha}$ for all $\alpha \in [0,1]$. Intuitively, this corresponds to the optimal response policies given the respective mean field $\boldsymbol{\mu}$.
$\Psi$ assigns $\boldsymbol{\pi} \in \boldsymbol{\Pi}$ to the $\boldsymbol{\mu} = \Psi (\boldsymbol{\pi}) \in \boldsymbol{\mathcal{M}}$ which is induced for all $\alpha \in [0,1]$ through the recursive equation
\begin{align}
	\mu_{t+1}^\alpha (x) = \sum_{x' \in \mathcal X} \mu_{t}^\alpha \bc{x'} \sum_{u \in \mathcal U} \pi_t^\alpha \bc{u \vert x'} P \bc{x \vert x', u, \mathbb G^\alpha_t}
\end{align}
with $\mu_0^\alpha = \mu_0$. Next, we state the concept of the colored digraphon mean field equilibrium (CDMFE) which extends existing work on MFGs, e.g. \cite{saldi2018markov}, to colored digraphons.
\begin{definition}
	A pair $\bc{\boldsymbol{\mu}, \boldsymbol{\pi}} \in \boldsymbol{\Pi} \times \boldsymbol{\mathcal M}$ is a CDMFE if and only if both $\boldsymbol{\pi} \in \Phi 
	\left(\boldsymbol{\mu}\right)$ and $\boldsymbol{\mu} = \Psi (\boldsymbol{\pi})$ hold.
\end{definition}
To ensure the existence of a CDMFE, we make a Lipschitz assumption commonly used in the literature.
\begin{assumption} \label{ass:Lip}
$r$, $P$, $W^1, \ldots, W^k$ are Lipschitz continuous with Lipschitz constants $L_r, L_P, L_{W_1}, \ldots L_{W_k} > 0$.
\end{assumption}

\begin{lemma}\label{lem:existence}
Under Assumption~\ref{ass:Lip} there exists a CDMFE.
\end{lemma}

\begin{proof}
We consider the extended state space $\mathcal{X} \times [0,1]$ which contains both the usual state $x \in \mathcal{X}$ as well as the agent index $\alpha \in [0,1]$. Then, Lemma~\ref{lem:existence} can be deduced from \cite[Theorem~3.3]{saldi2018markov}, see also \cite[Proof of Theorem~1]{cui2021learning}.
\end{proof}

\section{Approximation via CDMFGs}

A crucial benefit of CDMFGs is that they yield an accurate approximation of finite agent models with sufficiently many agents. To relate the two model classes, we discretize a policy ensemble $\boldsymbol{\pi} \in \boldsymbol{\Pi}$ using the function $\Gamma_N \colon \boldsymbol{\Pi} \to \Pi^N$ defined by $\Gamma_N (\boldsymbol{\pi}) = (\pi_1^{\alpha_1}, \ldots, \pi^{\alpha_N}_N)$ with $\alpha_i = i/N$. Conversely, an $N$-agent policy $(\pi_1, \ldots, \pi_N) \in \Pi^N$ can be transferred to the CDMFG setup by defining $\boldsymbol{\mu}^N \in \boldsymbol{\mathcal{M}}$ and $\boldsymbol{\pi}^N \in \boldsymbol{\Pi}$ through $\mu_t^{N, \alpha} \coloneqq \sum_{i \in V_N} \boldsymbol{1}_{\cbc{\alpha \in (\frac{i-1}{N}, \frac{i}{N}]}} \cdot \delta_{X_t^i}$ and $\pi_t^{N, \alpha} \coloneqq \sum_{i \in V_N} \boldsymbol{1}_{\cbc{\alpha \in (\frac{i-1}{N}, \frac{i}{N}]}} \cdot \pi_t^i$ for all $(\alpha, t) \in \mathcal{T} \times \mathcal{I}$.

To abbreviate the following formal expressions, we define 
$\boldsymbol{\mu}_t (f) \coloneqq \int_{\mathcal{I}} \sum_{x \in \mathcal{X}} f(x,a) \mu_t^\alpha (x) \mathrm d \alpha$ for all functions $f: \mathcal{X} \times \mathcal{I} \to \mathbb{R}$. Furthermore, we introduce the policy ensemble $\Gamma_N (\boldsymbol{\pi}, i, \hat \pi) \in \boldsymbol{\Pi}$ for all $i \leq N$ and $\hat \pi \in \Pi$ which is equal to $\Gamma_N (\boldsymbol{\pi})$
except that $\pi_i^{\alpha_i}$ is substituted by $\hat \pi$.

\begin{theorem} \label{thm:muconv}
Assume that Assumptions~\ref{ass:W} and \ref{ass:pi_lipschitz} hold with $\boldsymbol \pi \in \boldsymbol \Pi$  and $\boldsymbol \mu = \Psi(\boldsymbol \pi)$. Then, we have for the $N$-agent policy $\Gamma_N (\boldsymbol{\pi}, i, \hat \pi) \in \Pi^N$, for all measurable  functions $f \colon \mathcal X \times \mathcal I \to \mathbb R$ uniformly bounded by some $M_f$ that for all $t \in \mathcal{T}$
\begin{align}
    &\E \left[ \left| \boldsymbol \mu^N_t(f) - \boldsymbol \mu_t(f) \right|  \right] \to 0 
\end{align}
uniformly over all $\hat \pi \in \Pi, i \in V_N$.
\end{theorem}

\begin{proof}
First, we define the operator $P^{\boldsymbol{\pi}}_{t, \boldsymbol{\mu}', \mathbf{W}}$ as 
\begin{align*}
    &(\boldsymbol{\mu} P^{\boldsymbol{\pi}}_{t, \boldsymbol{\mu}', \mathbf{W}})^\alpha \coloneqq  \\ &\qquad \sum_{(x, u) \in \mathcal{X} \times \mathcal{U}} \mu^\alpha (x) \pi_t^\alpha (u \mid x) P \bc{\cdot \mid x, u, \mathbb{G}_t^{\alpha, \boldsymbol{\mu}', \mathbf{W}}}
\end{align*}
and $(W \mu) (\alpha, x) \coloneqq \int_\mathcal{I} W (\alpha, \beta) \mu^\beta (x) \mathrm d \beta$ and similarly $(\mu W)  (\alpha, x) \coloneqq \int_\mathcal{I} W (\beta, \alpha) \mu^\beta (x) \mathrm d \beta$ with $(W \mu - W' \mu') (\alpha, x) \coloneqq (W \mu) (\alpha, x) - (W' \mu') (\alpha, x)$ and analogously for $(\mu W  - \mu' W' )$.

We prove the desired statement via induction similar to \cite[Theorem 2]{cui2021learning}. The induction start follows by a law of large numbers argument. For the induction step we reformulate the term of interest, i.e.
\begin{align*}
    &\E \left[ \left| \boldsymbol \mu^N_{t+1}(f) - \boldsymbol \mu_{t+1}(f) \right| \right] \\
    &\quad \leq \E \left[ \left| \boldsymbol \mu^N_{t+1}(f) - \boldsymbol \mu^N_t P^{\boldsymbol \pi^N}_{t, \boldsymbol \mu^N_t, W_N}(f) \right| \right] \\
    &\qquad + \E \left[ \left| \boldsymbol \mu^N_t P^{\boldsymbol \pi^N}_{t, \boldsymbol \mu^N_t, W_N}(f) - \boldsymbol \mu^N_t P^{\boldsymbol \pi^N}_{t, \boldsymbol \mu^N_t, W}(f) \right| \right] \\
    &\qquad + \E \left[ \left| \boldsymbol \mu^N_t P^{\boldsymbol \pi^N}_{t, \boldsymbol \mu^N_t, W}(f) - \boldsymbol \mu^N_t P^{\boldsymbol \pi}_{t, \boldsymbol \mu^N_t, W}(f) \right| \right] \\
    &\qquad + \E \left[ \left| \boldsymbol \mu^N_t P^{\boldsymbol \pi}_{t, \boldsymbol \mu^N_t, W}(f) - \boldsymbol \mu^N_t P^{\boldsymbol \pi}_{t, \boldsymbol \mu_t, W}(f) \right| \right] \\
    &\qquad + \E \left[ \left| \boldsymbol \mu^N_t P^{\boldsymbol \pi}_{t, \boldsymbol \mu_t, W}(f) - \boldsymbol \mu_{t+1}(f) \right| \right] \, .
\end{align*}
By arguments analogous to \cite{cui2021learning}, the first, third, and fifth summand can be bounded by expressions that uniformly converge to zero as $N \to \infty$. To achieve similar statements for the second and fourth term for the CDMFG setup, however, we require an additional mathematical effort. We start with the second summand. Leveraging Assumption~\ref{ass:W} and keeping in mind that $ 0\leq \mu^{N,\beta}_t(x) \leq 1$ by definition, we obtain
\begin{align*}
    &\E \left[ \left| \boldsymbol \mu^N_t P^{\boldsymbol \pi^N}_{t, \boldsymbol \mu^N_t, W_N}(f) - \boldsymbol \mu^N_t P^{\boldsymbol \pi^N}_{t, \boldsymbol \mu^N_t, W}(f) \right| \right] \qquad \qquad \qquad
\end{align*}
    \vspace*{-0.35cm}
    \begin{align*}
    &\quad \leq \xi \E \left[ \int_{\mathcal I} \max_{h \in \cbc{1, \ldots, k} } \left\Vert  (W_N^h \mu^{N}_t - W^h \mu^{N}_t) (\alpha) \right\Vert  \right. \\
    &\qquad  +   \left. \max_{h \in \cbc{1, \ldots, k} } \left\Vert (\mu^{N}_t W_N^h - \mu^{N}_t W^h)(\alpha)  \right\Vert  \, \mathrm d\alpha \right] \\
    &\quad \leq \xi |\mathcal X|\sup_{x \in \mathcal X} \E \left[ \int_{\mathcal I} \max_{h \in \cbc{1, \ldots, k} }  \right. \left\vert (W_N^h \mu^{N}_t - W^h \mu^{N}_t) (\alpha, x)  \right\vert   \\
    &\qquad + \left. \max_{h \in \cbc{1, \ldots, k} } \left\vert ( \mu^{N}_t W_N^h - \mu^{N,\beta}_t W^h) (\alpha, x) \right\vert \, \mathrm d\alpha \right] \to 0
\end{align*}
where $\xi \coloneqq k |\mathcal X| M_f L_P$. We make a useful observation by defining the uniformly bounded functions $f'_{x', \alpha}(x, \beta) = W^h(\alpha, \beta) \cdot \mathbf 1_{x=x'}$ for any $(x', \alpha, h) \in \mathcal X \times \mathcal I \times \{1, \ldots k \}$. Applying the induction assumption to them yields
\begin{align*}
    &\E \left[ \int_{\mathcal I} \left|  (W^h \mu^{N}_t - W^h \mu_t) (\alpha, x')  \right| \, \mathrm d\alpha \right]\\
    &\qquad = \int_{\mathcal I} \E \left[ \left| \boldsymbol \mu^N_t(f'_{x', \alpha}) - \boldsymbol \mu_t(f'_{x', \alpha}) \right| \right] \, \mathrm d\alpha \to 0.
\end{align*}
Note that an analogous argument yields 
$ \E \left[ \int_{\mathcal I} \left|  (\mu^{N}_t W^h - \mu_t W^h) (\alpha, x')  \right| \, \mathrm d\alpha \right] \to 0$.
Then, we can prove convergence of the fourth summand, i.e.
\begin{align*}
    &\E \left[ \left| \boldsymbol \mu^N_t P^{\boldsymbol \pi}_{t, \boldsymbol \mu^N_t, \mathbf{W}}(f) - \boldsymbol \mu^N_t P^{\boldsymbol \pi}_{t, \boldsymbol \mu_t, \mathbf{W}}(f) \right| \right] \\
    &\quad \leq  |\mathcal X| M_f \E \left[ \sup_{x,u} \int_{\mathcal I} \left| P \left(x' \mid x, u, \mathbb{G}_t^{\alpha, \boldsymbol{\mu}_t^N, \mathbf{W}} \right) \right. \right. \\
    &\left.\left. \qquad  - P \left(x' \mid x, u, \mathbb{G}_t^{\alpha, \boldsymbol{\mu}_t, \mathbf{W}} \right) \right| \, \mathrm d\alpha \right] \\
    &\quad \leq \xi |\mathcal X| \sup_{x \in \mathcal X} \E \left[ \int_{\mathcal I}  \max_{h \in \cbc{1, \ldots, k} } \right. \left\vert  (W^h \mu^{N}_t- W^h \mu_t) (\alpha, x)  \right\vert\\
    &\qquad +   \max_{h \in \cbc{1, \ldots, k} } \left\vert (\mu^{N}_t W^h)(\alpha, x) \left. (\mu_t W^h(\alpha, x))  \right\vert \, \mathrm d\alpha \right] \to 0.
\end{align*}
This concludes the induction and the proof.
\end{proof}

Theorem~\ref{thm:muconv} can be leveraged to show that an equilibrium in a CDMFG 
yields an approximate equilibrium in the corresponding finite agent model. Moreover, the estimation's accuracy increases with the number of agents in the finite model which is formalized by the next theorem.
\begin{theorem} \label{thm:approxnash}
Let $(\boldsymbol \pi, \boldsymbol \mu)$ be a CDMFE and assume that Assumptions~\ref{ass:W}, \ref{ass:pi_lipschitz}, and \ref{ass:Lip} hold. Then, for any $\varepsilon, p > 0$ there exists $N'$ such that for all $N > N'$, the policy $\Gamma_N(\boldsymbol \pi) \in \Pi^N$ is an $(\varepsilon, p)$-Markov Nash equilibrium.
\end{theorem}

\begin{proof}
The result follows from Corollary~\ref{coro:jconv} in the Appendix.
\end{proof}

\section{Examples}
To verify our model numerically, we consider a number of example tasks fulfilling Assumption~\ref{ass:Lip}. For our examples, we consider the unweighted neighborhood state distributions
\begin{align}
	\mathbb G^i_{t, \textrm{out}} &\coloneqq \frac 1 {N} \sum_{h = 1}^k c_{h, \textrm{out}} \sum_{j \in V_N} \boldsymbol{1}_{\cbc{i j \textrm{ has color } h}} \delta_{X^j_t} \label{eq:simplifiedGout} \\
	\mathbb G^i_{t, \textrm{in}} &\coloneqq \frac 1 {N} \sum_{h = 1}^k  c_{h, \textrm{in}} \sum_{j \in V_N} \boldsymbol{1}_{\cbc{j i \textrm{ has color } h}} \delta_{X^j_t}\label{eq:simplifiedGin} 
\end{align}
where the coefficients $c_{h, \textrm{out}}$ and $c_{h, \textrm{in}}$ can be understood as the strength of the respective edge color. 

\paragraph{Epidemics}
For the epidemiological SIS problem from \cite{cui2021approximately} where agents may protect themselves ($D$) from costly infection ($I$), we endow the model with additional directionality of virus spread along directed edges. Therefore, formally $\mathcal X = \{S, I\}$, $\mathcal U = \{\bar D, D\}$, $\mu_0(I) = 0.5$, $T = 50$, $r \bc{X_{t}^i, U_{t}^i, \mathbb G^i_t} = - 2 \cdot \mathbf 1_{\{X_{t}^i = I\}} - 0.5 \cdot \mathbf 1_{\{U_{t}^i = D\}}$ and $P(S \mid I, U_{t}^i, \mathbb G^i_t) = 0.2$, $P(I \mid S, U_{t}^i, \mathbb G^i_t) = 0.8 \mathbb G^i_{t, \textrm{in}}(I) \cdot \mathbf 1_{\{U_{t}^i = \bar D\}}$.

\paragraph{Beach}
In the adjusted Beach example (e.g. \cite{perrin2020fictitious}), many people on the beach adjust their position to be both close to a bar and away from (incoming) disliked neighbors. There are $10$ locations $\mathcal X = \{0, 1, \ldots, 9\}$ with bar location $B = 5$, movement actions $\mathcal U = \{-1, 0, 1\}$, $\mu_0 = \mathrm{Unif}(\mathcal X)$, $T = 30$, $r \bc{X_{t}^i, U_{t}^i, \mathbb G^i_t} = \frac{1}{5} \abs{B - X_{t}^i} + \frac{1}{5} \abs{U_{t}^i} - \mathbb G^i_{t, \textrm{in}} (x)$ and $X_{t+1}^i = X_{t}^i + U_{t}^i + \epsilon_t^i$, where the random noise $\epsilon_t^i$ is equal to $-1$ or $1$ with probabilities $0.05$, and otherwise $0$.

\paragraph{Systemic Risk}
In the context of an interbanking loan market, the weights $c_{h, \textrm{out}}$, $c_{h, \textrm{in}}$ indicate the amount of money one financial institution lends to another one.
If a bank is in state $x \in \mathcal X \coloneqq \cbc{h, \ell, b}$, its liquidity is defined as 
\begin{align*}
    c_x \coloneqq  \xi_x
    + \left( \sum_{x' \in \mathcal{X}} \mathbb G^i_{t, \textrm{in}} (x') \cdot w_{x'} - \mathbb G^i_{t, \textrm{out}} (x') \cdot (3 - w_{x'}) \right)
\end{align*}
where each state $x$ is assigned an initial level of capital endowment $\xi_x \in \mathbb{R}$. For notational convenience we  define $\sigma_{c, x} \coloneqq \max \{-1, \min\{1, c_x \}\}$ for all states $x \in \mathcal{X}$. Banks can choose to keep their capital at it's current level, raise or decrease it which is formalized by $\mathcal{U} \coloneqq \{k, r, d \}$.
We model negative external shocks, e.g. wars, by including a parameter $1 > \beta > 0$ in the transition probabilities. Intuitively, a bank is more likely to be in a high liquidity state in $t+1$ if it already is there in $t$ or if it decides to raise it's liquidity.

Thus, for $(i, j) \in \{h, \ell, b \}^2$ define a $3 \times 3$ transition matrix $M (\alpha, t, k)$ for choosing action $u = k$ by  $M_{i, h} (\alpha, t, k) \coloneqq (1 + \sigma_{c, i})/4$ and $M_{i, \ell} (\alpha, t, k) \coloneqq (1 - 0.25 \sigma_{c, i} + \beta \cdot \boldsymbol{1}_{\{ i = h \}})/4$ and $M_{i, b} (\alpha, t, k) \coloneqq (2 - 0.75 \sigma_{c, i} - \beta \cdot \boldsymbol{1}_{\{ i = h \}})/4$.

If a bank chooses to raise it's capital, the respective transition matrix $M (\alpha, t, r)$ is defined by $M_{h, j} (\alpha, t, r) \coloneqq M_{h, j} (\alpha, t, k)$ and $M_{\ell, j} (\alpha, t, r) \coloneqq (1 - \lambda) M_{\ell, j} (\alpha, t, k) + \lambda( M_{h, j} (\alpha, t, k) - \beta \boldsymbol{1}_{\{ j = \ell \}} + \beta \boldsymbol{1}_{\{ j = b \}})$ and $M_{b, j} (\alpha, t, r) \coloneqq (1 - \lambda) M_{b, j} (\alpha, t, k) + \lambda M_{\ell, j} (\alpha, t, k)$.

Finally, if a bank implements action $d$, the transition matrix $M (\alpha, t, d)$ is defined by $M_{h, j} (\alpha, t, d) \coloneqq (1 - \lambda) M_{h, j} (\alpha, t, k) + \lambda (M_{\ell, j} (\alpha, t, d) + \beta \boldsymbol{1}_{\{ j = \ell \}} - \beta \boldsymbol{1}_{\{ j = b \}})$ and $M_{\ell, j} (\alpha, t, d) \coloneqq (1 - \lambda) M_{\ell, j} (\alpha, t, k) + \lambda M_{b, j} (\alpha, t, k)$ and $M_{b, j} (\alpha, t, d) \coloneqq M_{b, j} (\alpha, t, k)$.

Banks try to maximize	$J \bc{u} := \Erw \brk{ - \sum_{t = 0}^T k_h \mathbf{1}_h + k_b \mathbf{1}_b}$
where $k_b \gg k_h$ are the costs for being in the respective state. In our experiments, we use $\beta = 0.6$, $\mu_0(h) = \mu_0(l) = 0.5$, $T=50$, $\lambda = 0.1$, $k_b = 10k_h = 0.4$, $(\xi_h, \xi_l, \xi_b) = (1, 0, -1)$ and $(w_h, w_l, w_b) = (1, 0, -1)$.

\paragraph{Adaptive versions}
We also introduce adaptive versions of SIS and Systemic Risk by changing the underlying graph structure (for simplicity at $t^* = T/2$) via time-dependent coefficients $c_{h, \textrm{out}} = c_{h, \textrm{in}} = \mathbf 1_{\{t \in ((h-2)T/2, (h-1)T/2] \}}$ in \eqref{eq:simplifiedGout} and \eqref{eq:simplifiedGin}. For adaptive SIS, we adjust $P(I \mid S, U_{t}^i, \mathbb G^i_t) = \min(1, 1.6 \mathbb G^i_t(I) \cdot \mathbf 1_{\{U_{t}^i = \bar D\}})$.

\begin{figure}[tb]
    \centering
    \vspace{0.2cm}
    \includegraphics[width=0.95\linewidth]{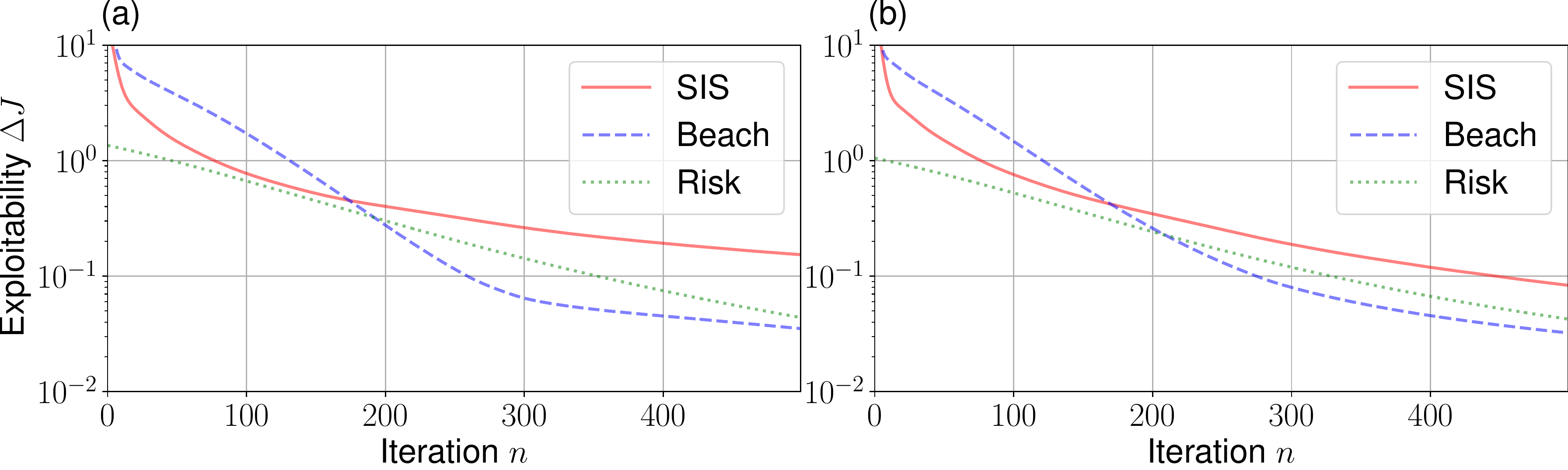}
    \caption{The (approximate) exploitability of policies $\Delta J(\pi) \coloneqq \int_{\mathcal I} \sup_{\pi^* \in \Pi} J^{\Psi(\pi)}_\alpha(\pi^*) - J^{\Psi(\pi)}_\alpha(\pi) \, \mathrm d\alpha$ at each iteration $n$ of the online mirror descent algorithm is decreasing. (a): The rotated uniform attachment digraphon scenario; (b): The ranked attachment version. }
    \label{fig:expl}
\end{figure}

\begin{figure}[tb]
    \centering
    \includegraphics[width=0.95\linewidth]{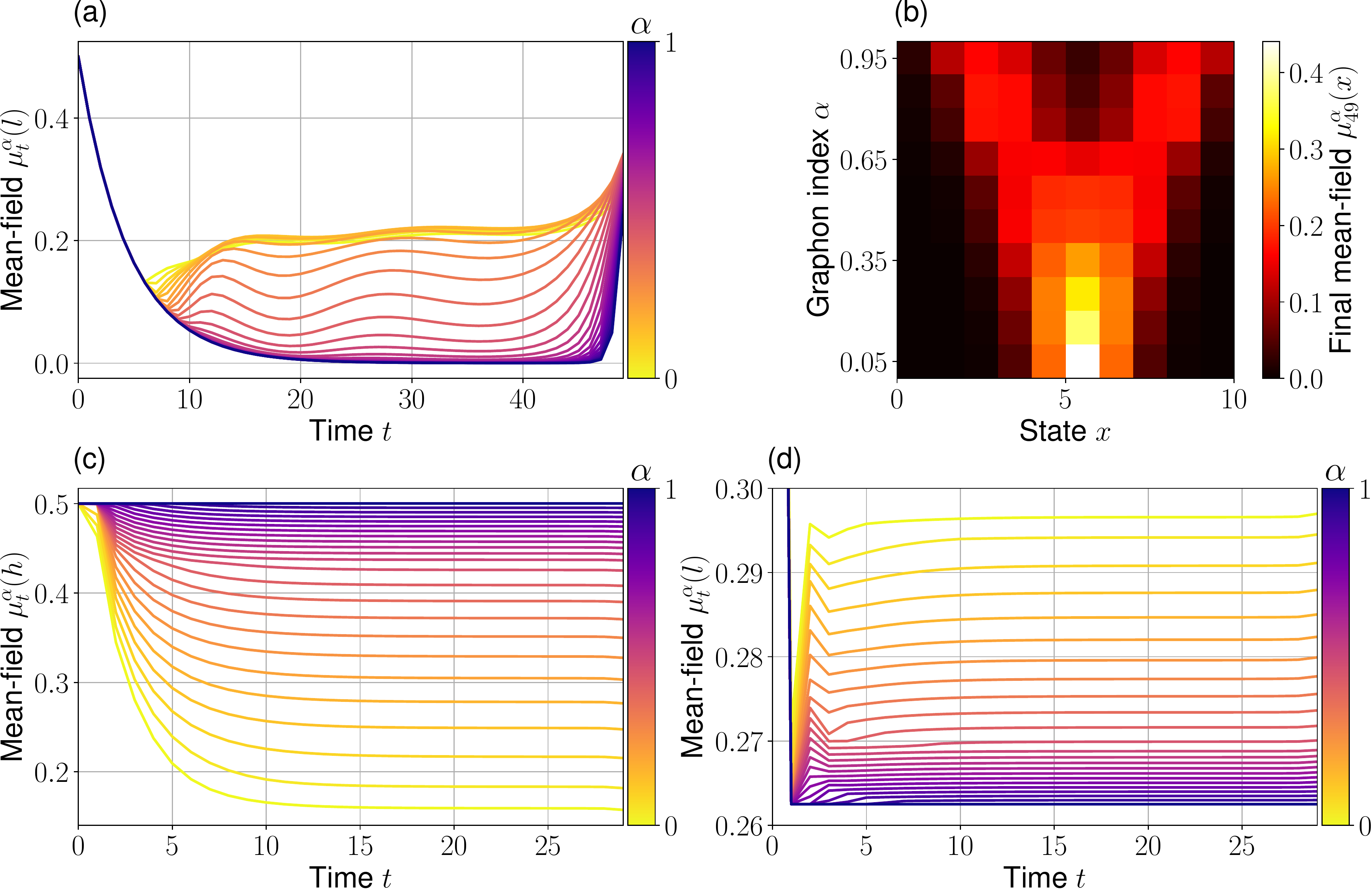}
    \caption{The limiting mean-field (state marginals) for various graphon indices $\alpha \in [0,1]$ in the combined uniform ranked attachment scenario. (a): SIS; (b): Beach; (c-d): Systemic Risk.}
    \label{fig:mf}
\end{figure}

\begin{figure}[tb]
    \centering
    \vspace{0.05cm}
    \includegraphics[width=0.95\linewidth]{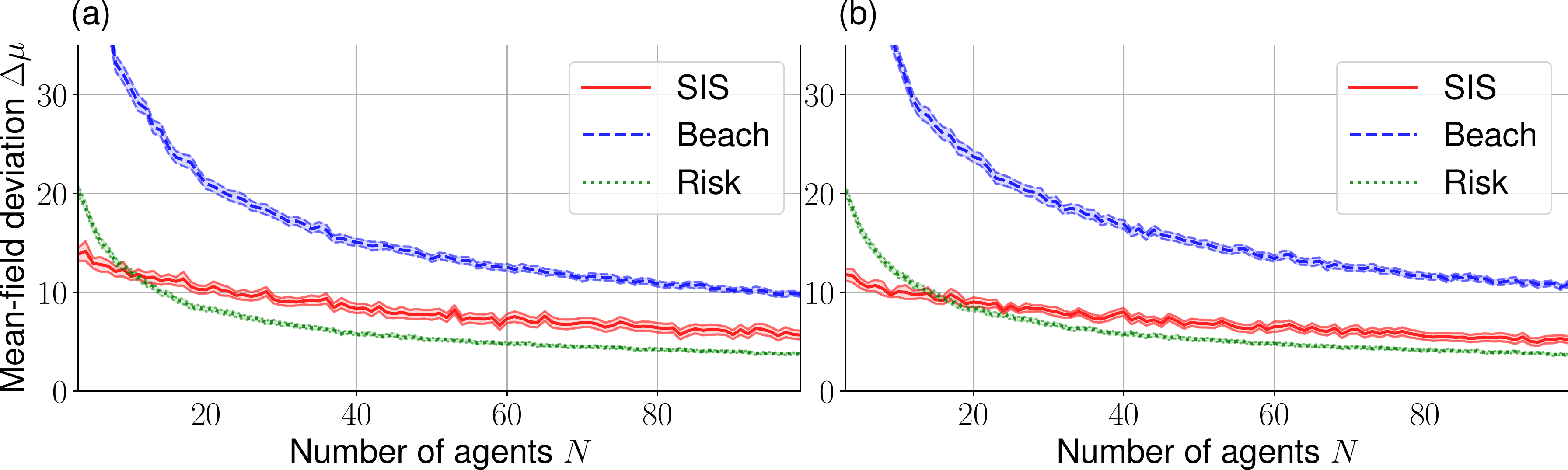}
    \caption{The estimated deviation between empirical and limiting mean field $\Delta \mu \coloneqq \mathbb E \left[ \sum_{t \in \mathcal T, x \in \mathcal X} \left| \frac 1 N \sum_{i} \delta_{X^i_t}(x) - \int_{\mathcal I} \mu^\alpha_t(x) \, \mathrm d\alpha \right| \right]$ and its $95\%$ confidence interval over the number of nodes $N$, for each averaged over $100$ random samples. (a): The rotated uniform attachment digraphon scenario; (b): The ranked attachment version.}
    \label{fig:conv-mf}
\end{figure}

\begin{figure}[tb]
    \centering
    \vspace{0.2cm}
    \includegraphics[width=0.95\linewidth]{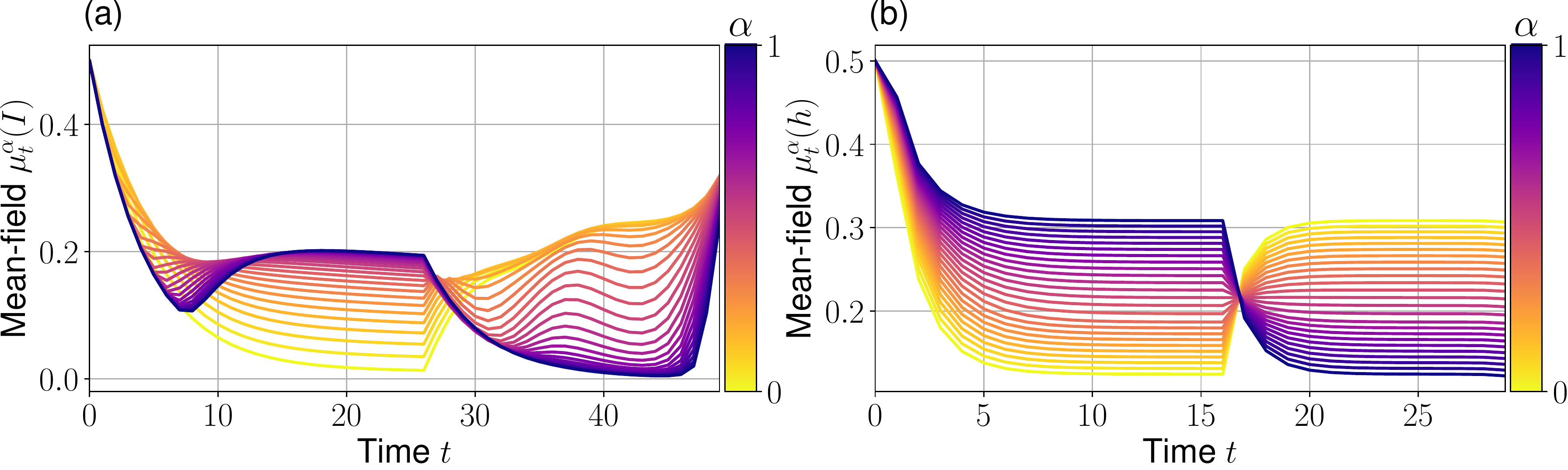}
    \caption{The limiting mean-field (state marginals) for various graphon indices $\alpha \in [0,1]$ in the adaptive graph case. (a): The SIS example in the combined uniform ranked attachment digraphon scenario; (b): The Systemic Risk example in the double rotated uniform attachment digraphon scenario.}
    \label{fig:adaptive}
\end{figure}

\section{Experiments}
To verify the usefulness of our proposed models, we use the online mirror descent algorithm \cite{perolat2021scaling} and discretization of the graphon index space $[0,1]$ \cite{cui2021learning} to numerically find CDMFE for the aforementioned examples and the following digraphons. We consider the rotated uniform attachment digraphon scenario resulting from the uniform attachment process (see e.g. \cite{lovasz2012large}) with added directionality
\begin{align*}
    W^1 = 1-W^2, W^2(x,y) = W_{\mathrm{r-unif}}(x,y) \coloneqq 1 - x(1-y),
\end{align*}
the double rotated uniform attachment digraphon scenario
\begin{align*}
    W^1 = 1-W^2-W^3, W^2 = \frac 1 2 W_{\mathrm{r-unif}}, W^3 = \frac 1 2 W_{\mathrm{l-unif}}
\end{align*}
where $W_{\mathrm{l-unif}}(x,y) \coloneqq 1 - (1-x)y$, and the combined uniform ranked attachment digraphon scenario
\begin{align*}
    W^1 = 1-W^2-W^3, W^2 = \frac 1 2 W_{\mathrm{r-unif}}, W^3 = \frac 1 2 W_{\mathrm{r-rank}}
\end{align*}
where $W_{\mathrm{r-rank}} \coloneqq 1 - \max(x,(1-y))$, with weights $c_{h, \textrm{out}} = c_{h, \textrm{in}} = h-1$ unless stated otherwise.

As can be seen in Fig.~\ref{fig:expl}, the numerical approach converges to zero in approximate exploitability \cite{cui2021learning}, indicating good convergence to a CDMFE. The resulting behavior is intuitive as shown in Fig.~\ref{fig:mf}: For the SIS problem, agents with less incoming edges (low $\alpha$) are less likely to protect themselves due to the lower perceived risk, leading to infections; In the Beach scenario, agents with more incoming connections will spread out further away from the bar to avoid other agents, while agents with little to no incoming connections will stay near the bar; Finally, in the systemic risk scenario, banks borrowing more money from other banks are more likely to keep low capital reserves. In Fig.~\ref{fig:conv-mf}, we can see that the difference between the empirical mean-field and the limiting mean-field converges to zero as the finite graph system becomes sufficiently big, supporting Theorem~\ref{thm:muconv}. Lastly, in Fig.~\ref{fig:adaptive}, for the adaptive scenario we can see analogous behavior to the one seen in Fig.~\ref{fig:mf} with qualitative reversals in behavior at $t=T/2$ when the set of active edges shifts.

\section{Conclusion}

In this paper we have constructed and analyzed CDMFGs by providing both rigorous mathematical results as well as illustrating our theory on different examples. CDMFGs extend the existing GMFG theory to weighted and directed graphs and thereby enable the scalable analysis of many problems that are outside the scope of standard GMFGs. Furthermore, CDMFGs yield a framework to depict adaptive graph structures, as we have demonstrated in the experiments. For future work, possible extensions of CDMFGs are manifold and could include partial observability, continuous time, state, and action spaces as well as using CDMFGs in applications. We hope that our work inspires future MFG research and facilitates the analysis of many real-world problems.

\bibliography{references}
\bibliographystyle{IEEEtran}

\section{Appendix}

Define $\forall t \in \mathcal T$ the dynamics with
\begin{align}
    \hat U^{\frac i N}_t \sim \hat \pi_t(\cdot \mid \hat X^{\frac i N}_t), \, \hat X^{\frac i N}_{t+1} \sim P(\cdot \mid \hat X^{\frac i N}_t, \hat U^{\frac i N}_t, \mathbb G^{\frac i N}_t)
\end{align}
and $\hat X^{\frac i N}_0 \sim \mu_0$ for all agents $i$.

\begin{lemma} \label{lem:xconv}
 Assume that Assumptions~\ref{ass:W}, \ref{ass:pi_lipschitz}, and \ref{ass:Lip} hold with $\boldsymbol \pi \in \boldsymbol \Pi$ and $\boldsymbol \mu = \Psi(\boldsymbol \pi)$. For $\Gamma_N(\boldsymbol \pi), \Gamma (\boldsymbol{\pi}, i, \hat \pi) \in \Pi^N$ with arbitrary $\hat \pi \in \Pi$, and for any uniformly bounded family $\mathcal G$ of functions $g: \mathcal X \to \mathbb R$ and any $\varepsilon, p > 0$, $t \in \mathcal T$, there exists $N' \in \mathbb N$ such that for all $N > N'$ we have
\begin{align} \label{eq:xconv}
    \sup_{g \in \mathcal G} \left| \E \left[ g(X^i_{t}) \right] - \E \left[ g(\hat X^{\frac i N}_{t}) \right] \right| < \varepsilon
\end{align}
uniformly over $\hat \pi \in \Pi, i \in V'_N$ for some $V'_N \subseteq V_N$ with $|V'_N| \geq \left\lfloor (1-p) N \right\rfloor$. Under the same conditions,
\begin{align} \label{eq:xmuconv}
    &\sup_{h \in \mathcal H} \left| \E \left[ h \left(X^i_{t}, \mathbb{G}_t^{\frac i N, \boldsymbol{\mu}_t^N, \mathbf{W_N}} \right) - h \left(\hat X^{\frac i N}_{t}, \mathbb{G}_t^{\frac i N, \boldsymbol{\mu}_t, \mathbf{W}} \right) \right] \right| < \varepsilon
\end{align}
holds for any uniformly Lipschitz, uniformly bounded family $\mathcal{H}$ of functions $h: \mathcal X \times \mathcal (B_1(\mathcal X))^{2 k} \to \mathbb R$.
\end{lemma}

\begin{proof}
The first statement \eqref{eq:xconv} can be proven with an argument in \cite[Proof of Lemma A.1]{cui2021learning}. To show that \eqref{eq:xmuconv} also holds, we will establish that \eqref{eq:xconv} implies \eqref{eq:xmuconv}.
Thus, let $\mathcal H$ be defined as in Lemma~\ref{lem:xconv} with uniform Lipschitz constant $L_h$ and bound $M_h$. For any $h \in \mathcal{H}$, we have
\begin{align*}
    &\left| \E \left[ h(X^i_{t}, \mathbb{G}_t^{\frac i N, \boldsymbol{\mu}_t^N, \mathbf{W}_N}) \right] - \E \left[ h(\hat X^{\frac i N}_{t}, \mathbb{G}_t^{\frac i N, \boldsymbol{\mu}_t, \mathbf{W}}) \right] \right| \\
    &\quad \leq \left| \E \left[ h(X^i_{t}, \mathbb{G}_t^{\frac i N, \boldsymbol{\mu}_t^N, \mathbf{W}_N}) \right] - \E \left[ h(X^i_{t}, \mathbb{G}_t^{\frac i N, \boldsymbol{\mu}_t, \mathbf{W}_N}) \right] \right| \\
    &\qquad + \left| \E \left[ h(X^i_{t}, \mathbb{G}_t^{\frac i N, \boldsymbol{\mu}_t, \mathbf{W}_N}) \right] - \E \left[ h(X^i_{t}, \mathbb{G}_t^{\frac i N, \boldsymbol{\mu}_t, \mathbf{W}}) \right] \right| \\
    &\qquad + \left| \E \left[ h(X^i_{t}, \mathbb{G}_t^{\frac i N, \boldsymbol{\mu}_t, \mathbf{W}}) \right] - \E \left[ h(\hat X^{\frac i N}_{t}, \mathbb{G}_t^{\frac i N, \boldsymbol{\mu}_t, \mathbf{W}}) \right] \right|.
\end{align*}
While the third term can be bounded as in \cite[Proof of Lemma A.1]{cui2021learning}, the first two summands converge to zero by the following argument. Starting with the first term, we obtain
\begin{align*}
    &\left| \E \left[ h(X^i_{t}, \mathbb{G}_t^{\frac i N, \boldsymbol{\mu}_t^N, \mathbf{W}_N}) \right] - \E \left[ h(X^i_{t}, \mathbb{G}_t^{\frac i N, \boldsymbol{\mu}_t, \mathbf{W}_N}) \right] \right| \\
    & \leq \E \left[ \E \left[ \left| h(X^i_{t}, \mathbb{G}_t^{\frac i N, \boldsymbol{\mu}_t^N, \mathbf{W}_N}) - h(X^i_{t}, \mathbb{G}_t^{\frac i N, \boldsymbol{\mu}_t, \mathbf{W}_N}) \right| \innermid X^i_{t} \right] \right] \\
    &\quad \leq k L_h \sum_{x \in \mathcal X} \E \left[  \max_{j \in \cbc{1, \ldots, k}}\left| (W^j_N \mu^{N}_t - W^j_N \mu_t) (\frac i N, x)  \right|  \right. \\
    &\qquad \left. +   \max_{j \in \cbc{1, \ldots, k}}\left| (\mu^{N}_t W^j_N - \mu_t W^j_N) (\frac i N, x)  \right|  \right]
    = o(1)
\end{align*}
where we apply Theorem~\ref{thm:muconv} to the uniformly bounded functions $f_{N,i,x}'(x', \beta) = W^j_N(\frac i N, \beta) \cdot \mathbf 1_{x=x'}$  and $f_{N,i,x}'' (x', \beta) = W^j_N(\beta, \frac i N) \cdot \mathbf 1_{x=x'}$. For the second summand, we have
\begin{align*}
    &\left| \E \left[ h(X^i_{t}, \mathbb{G}_t^{\frac i N, \boldsymbol{\mu}_t, \mathbf{W}_N}) \right] - \E \left[ h(X^i_{t}, \mathbb{G}_t^{\frac i N, \boldsymbol{\mu}_t, \mathbf{W}}) \right] \right| \\
    &\quad \leq k L_h \sum_{x \in \mathcal X}  \max_{j \in \cbc{1, \ldots, k}}\left|   (W^j_N\mu_t - W^j\mu_t)(\frac i N, x)   \right|\\
    &\qquad + k L_h \sum_{x \in \mathcal X}  \max_{j \in \cbc{1, \ldots, k}}\left|   (\mu_t W^j_N - \mu_t W^j)(\frac i N, x)   \right|
\end{align*}
where each of the two summands converges to zero by an argument as in \cite[Proof of Lemma A.1]{cui2021learning}.
\end{proof}

This result implies that the objective functions of almost all agents converge uniformly to the mean field objectives.

\begin{corollary} \label{coro:jconv}
Assume that Assumptions~\ref{ass:W}, \ref{ass:pi_lipschitz}, and \ref{ass:Lip} hold with $\boldsymbol \pi \in \boldsymbol \Pi$ and $\boldsymbol \mu = \Psi(\boldsymbol \pi)$. For $\Gamma_N(\boldsymbol \pi), \Gamma (\boldsymbol{\pi}, i, \hat \pi) \in \Pi^N$ with arbitrary $\hat \pi \in \Pi$, for any $\varepsilon, p > 0$, there exists $N' \in \mathbb N$ such that for all $N > N'$ we have
\begin{align} \label{eq:Jconv}
    \left| J_i^N(\pi^1, \ldots, \pi^{i-1}, \hat \pi, \pi^{i+1}, \ldots, \pi^N) - J^{\boldsymbol \mu}_{\frac i N}(\hat \pi) \right| < \varepsilon
\end{align}
uniformly over $\hat \pi \in \Pi, i \in V'_N$ for some $ V'_N \subseteq  V_N$ with $|V'_N| \geq \left\lfloor (1-p) N \right\rfloor$.
\end{corollary}

\begin{proof}
The proof is analogous to \cite[Proof of Corollary A.1]{cui2021learning}.
\end{proof}

\end{document}